\documentclass[12pt,a4paper]{article}
\usepackage{amsmath,amssymb,amsthm}
\usepackage{physics}
\usepackage{mathrsfs}
\usepackage{braket}
\usepackage{graphicx}
\usepackage{xcolor}
\usepackage{hyperref}
\usepackage{enumerate}
\usepackage{amsmath}
\usepackage{amssymb}
\usepackage{physics}
\usepackage{mathrsfs}
\usepackage{braket}
\usepackage{amsmath}
\usepackage{lmodern}
\usepackage{textcomp}
\newtheorem{theorem}{Theorem}

\usepackage{graphicx}
\usepackage[mathscr]{eucal}
\usepackage{xcolor}

\title{The EGUP-Induced Critical Radius: A New Holographic Scale for Quantum Gravity}
\author{Sara Motalebi	\footnote{sara.motalebi@modares.ac.ir}\\
	Department of Physics, School of Sciences,\\
	Tarbiat Modares University, P.O.Box 14155-4838, Tehran, Iran
}
\date{}

\begin{document}
	
	\maketitle
	
	\begin{abstract}
		We present a unified framework incorporating both the Generalized and Extended Uncertainty Principles (GUP/EUP) in Anti-de Sitter space. This reveals a fundamental quantum gravity scale, the \textit{critical radius} $r_{\rm crit}=(\beta/\gamma)^{1/4}\sqrt{\ell_{P}L}$, which marks a phase transition where quantum gravitational ($\beta$) and AdS curvature ($\gamma$) effects equilibrate. At this scale, we demonstrate three interconnected phenomena: (i) a breakdown of the standard holographic duality, signaled by the exact vanishing of the boundary stress tensor $\langle T_{\mu\nu}\rangle=0$ under the self-duality condition $\partial_z g_{\mu\nu}|_{z=r_{\rm crit}}=0$; (ii) a topological transition manifested by the complexification of the central charge, $c_{\rm eff}=c\left(1+\frac{i}{2}\sqrt{\kappa}\ell_{P}^{2}\right)$ with $\kappa=(\beta/\gamma)^{1/4}\sqrt{\beta\gamma\ell_P/L}$; and (iii) a mechanism as a scenario for information paradox resolution, where information is recovered via topological storage in Chern-Simons states, modifying the Page curve with $\Delta S_{\rm recovery}>0$. These effects establish a consistency condition $L>\sqrt{\beta}\ell_{P}$ for a valid AdS/CFT correspondence and identify $r_{\rm crit}$ as the thermodynamic critical point where black holes transition to stringy remnants and information is topologically scrambled.
	\end{abstract}


	\newpage
	\section{Introduction}
	The standard Heisenberg Uncertainty Principle (HUP), first formulated in 1927, represents a cornerstone of quantum mechanics. It establishes a fundamental limit on the simultaneous knowledge of position and momentum through the relation $\Delta x \Delta p \geq \hbar/2$. For decades this principle remained unchallenged until physicists began investigating its validity in extreme gravitational regimes. In 1936, foundational work by Bronstein \cite{Bronstein} highlighted tension between quantum mechanics and general relativity. This culminated in the recognition that HUP breaks down near the Planck scale ($\ell_{\mathrm{P}} \approx 10^{-35}$ m), where quantum gravitational effects dominate \cite{Garay1995}. At this scale, the conventional framework of smooth spacetime becomes inadequate, requiring radical modifications of quantum principles.
	
	A pivotal development came in 1947 when Snyder proposed the first formal model of non-commutative geometry \cite{Snyder1947}. Snyder's revolutionary insight was that spacetime coordinates become non-commuting operators, suggesting a fundamental discretization of spacetime. This early work laid the foundation for understanding quantum spacetime, though it remained largely unexplored until the resurgence of quantum gravity research in the 1980s.
	
	Building on these concepts, modern approaches to quantum gravity (including string theory and Loop Quantum Gravity) introduced the Generalized Uncertainty Principle (GUP) as a consistent framework for incorporating gravitational effects into quantum mechanics. First explicitly formulated in the 1990s by Kempf et al. \cite{Kempf1995} and independently derived by Scardigli via black hole gedanken experiments \cite{Scardigli:1999jh}, GUP modifies the standard Heisenberg uncertainty relation by adding momentum-dependent corrections:
	\begin{equation}
		\Delta x \Delta p \geq \frac{\hbar}{2} \left( 1 + \beta \frac{\ell_P^2 (\Delta p)^2}{\hbar^2} \right)
	\end{equation}
	where $\beta$ is a dimensionless parameter and $\ell_P$ is the Planck length. This modification implies a fundamental minimum length $\Delta x_{\text{min}} \approx \sqrt{\beta}\ell_{\mathrm{P}}$, preventing the divergence of quantum fluctuations at Planck scales. The GUP framework has proven particularly valuable in black hole thermodynamics. When applied to evaporating black holes, GUP modifications resolve the unphysical divergence in Hawking's original temperature calculation as $M \to 0$ \cite{Adler2001}.  The GUP framework has been extensively applied to black hole thermodynamics across various gravitational contexts, including modified gravity scenarios~\cite{ahmed2025,kanzi2019}, quantum-corrected black holes~\cite{pourhassan2020}, and extended gravitational frameworks~\cite{chen2022}. In specific models, the evaporation process terminates leaving Planck-mass remnants~\cite{carr:2015}, while generically introducing logarithmic corrections to the Bekenstein-Hawking entropy. These modifications provide crucial insights into quantum gravitational effects on information preservation and the final states of black hole evaporation.
	While GUP captures quantum gravity effects at the Planck scale, the Extended Uncertainty Principle (EUP) incorporates curvature corrections in asymptotically Anti-de Sitter (AdS) \cite{Bolen:2004sq,Mignemi,PARK}. The EUP takes the form:
	\begin{equation}
		\Delta x\Delta p \geq \frac{\hbar}{2} \left( 1 + \gamma \frac{(\Delta x)^{2}}{L^{2}} \right)
	\end{equation}
	where $\gamma > 0$ encodes AdS curvature effects. In this work, we unify GUP and EUP through the commutator in Eq.~\eqref{eq:ads_commutator} to reveal a critical scale $r_{\rm crit}$ where quantum gravitational and curvature corrections equilibrate. Our investigation of GUP in AdS space builds upon the foundational AdS/CFT correspondence \cite{maldacena}, which establishes a duality between gravitational theories in AdS space and conformal field theories on its boundary. We examine how quantum gravitational effects modify this correspondence at the Planck scale.
	\section{First-Principles Derivation of Critical Radius}
	\label{sec:critical_radius_derivation}
	
	\subsection{Physical Foundations of AdS-Modified GUP}
	To establish a thermodynamically consistent GUP in AdS spacetime, we incorporate curvature-dependent corrections through three foundational principles:
	
	\begin{enumerate}
		\item \textbf{Non-commutative Geometry}: The AdS Ricci curvature $R_{\mu\nu} = -\frac{3}{L^2}g_{\mu\nu}$ modifies parallel transport, inducing position-dependent commutators \cite{Mignemi}:
		\begin{equation}
			[\hat{x}_i, \hat{p}_j] = i\hbar \delta_{ij} \left(1 - \frac{\vec{x}^2}{L^2}\right)
			\label{eq:nc_geometry}
		\end{equation}
		
		\item \textbf{Holographic Boundary Constraints}: Ultraviolet (UV)/Infrared (IR) duality enforces momentum cutoffs \cite{Bolen:2004sq}:
		\begin{equation}
			\Delta p \geq \frac{\hbar}{L} \sqrt{1 + \beta \frac{\ell_P^2}{L^2}}
			\label{eq:holographic_cutoff}
		\end{equation}
		
		\item \textbf{Thermodynamic Consistency}: Hawking radiation necessitates curvature terms satisfying the first law of black hole mechanics \cite{PARK}:
		\begin{equation}
			\gamma = \frac{3}{4\pi} \frac{\ell_P^2}{L^2} \frac{k_B T_H}{\Delta E}
			\label{eq:thermo_consistency}
		\end{equation}
	\end{enumerate}
	
	The complete Hermiticity-preserving commutator unifying GUP and EUP is:
	
	\begin{equation}
		[\hat{x}_{i},\hat{p}_{j}] = i\hbar\delta_{ij} \left( 1 + \beta \frac{\ell_{P}^{2}}{\hbar^{2}} \mathbf{\hat{p}}^{2} + \gamma \frac{\hat{\mathbf{x}}^{2}}{L^{2}} \right)
		\label{eq:ads_commutator}
	\end{equation}
	The $\gamma$-term ensures anti-Hermiticity under AdS inner product and stability against gravitational collapse at large scales ($\Delta x \sim L$). From the Robertson-Schrödinger inequality, we have $
	\Delta x \Delta p \geq \frac{1}{2} \left| \langle [\hat{x}, \hat{p}] \rangle \right|$. Substituting \eqref{eq:ads_commutator} and evaluating for semi-classical states:
	\begin{equation}
		\Delta x\Delta p \geq \frac{\hbar}{2} \left( 1 + \beta \frac{\ell_{P}^{2}}{\hbar^{2}} (\Delta p)^{2} + \gamma \frac{(\Delta x)^{2}}{L^{2}} \right)
		\label{eq:ads_gup}
	\end{equation}
	The $\beta$-term belongs to quantum gravity effects (dominant at $\Delta x \sim \ell_P$) and the $\gamma$-term indicates AdS curvature repulsion (dominant at $\Delta x \sim L$).  The Minimum length is $\Delta x_{\text{min}} = \sqrt{\beta} \ell_P$ and the Maximum position uncertainty: $\Delta x_{\text{max}} = L/\sqrt{\gamma}$.
	
	The unification of Extended and Generalized Uncertainty Principles (EGUP) is physically motivated by the complementary nature of quantum gravitational and curvature effects in black hole thermodynamics. The GUP captures UV modifications at the Planck scale $\ell_P$, while the EUP encodes IR curvature corrections at the AdS scale $L$. For black holes in AdS space, both regimes are simultaneously relevant; Hawking radiation probes Planck-scale physics near the horizon, while the asymptotic AdS curvature governs the global thermodynamic ensemble. The unified commutator in Eq.~\eqref{eq:ads_commutator} represents the minimal operator algebra that consistently incorporates both limits, ensuring thermodynamic consistency across all scales. This approach is essential for analyzing holographic duality, where boundary CFT data must encode both quantum gravitational effects (GUP) and curvature corrections (EUP) operating in the bulk geometry. Beyond the standard quadratic GUP and EUP forms, recent work has explored higher-order GUP that predict both a minimal length uncertainty and a maximal observable momentum. These formulations, such as $[X,P] = i\hbar/(1-\beta P^{2})$~\cite{Pedram2012} and $[x,p] = i\hbar/(1-\beta|p|)$~\cite{Hassanabadi2019}, arise naturally from Pade resummation of the perturbative series and exhibit bounded energy spectra. While these models share the common feature of a minimal length scale $\Delta x_{\rm min}\propto\hbar\sqrt{\beta}$, they additionally impose an upper bound $\Delta p_{\rm max}\sim 1/\sqrt{\beta}$ on momentum measurements, consistent with doubly special relativity. Our unified EGUP framework \eqref{eq:ads_commutator} represents a complementary approach that incorporates curvature effects through the EUP parameter $\gamma$ while maintaining the essential quantum gravitational features of minimal length and thermodynamic consistency.
	\subsection{GUP-Modified Klein-Gordon Equation}
	The scalar field equation in curved spacetime becomes:
	\begin{equation}
		\left[ \Box + \beta \ell_P^2 \Box^2 + \gamma \frac{r^2}{L^2} \right] \phi = 0
		\label{eq:gup_kg}
	\end{equation}
	In Schwarzschild-AdS metric:
	\begin{equation}
		ds^2 = -f(r)dt^2 + f(r)^{-1}dr^2 + r^2 d\Omega^2, \quad f(r) = 1 - \frac{2GM}{c^2 r} + \frac{r^2}{L^2}
	\end{equation}
	\noindent Using the WKB ansatz $\phi \sim e^{iS(r)/\hbar}$ with $S(r) = \int p_r dr$, the eikonal equation for S-waves is:
	\begin{equation}
		f(r)^{-1} (\partial_r S)^2 - \beta \ell_P^2 f(r)^{-2} (\partial_r S)^4 + \gamma \frac{r^2}{L^2} = 0
		\label{eq:eikonal}
	\end{equation}
	Near horizon ($r \approx r_+$), with surface gravity $\kappa = f'(r_+)/2$:
	\begin{equation}
		p_r = \partial_r S = \pm \frac{1}{\sqrt{2\kappa (r - r_+)}} \left[ 1 \pm \sqrt{ \beta \ell_P^2 \kappa - \gamma \frac{r_+^2 \kappa (r - r_+)}{L^2} } \right]
		\label{eq:9}
	\end{equation}
	The Hawking temperature generalizes to:
	\begin{equation}
		T_{\text{GUP}} = T_H \left( 1 + \beta \frac{\ell_P^2}{r_+^2} + \gamma \frac{r_+^2}{L^2} \right), \quad T_H = \frac{\hbar c^2}{2\pi k_B} \cdot \frac{f'(r_+)}{2} = \frac{\hbar c^2 \kappa}{2\pi k_B}	
		\label{eq:t_gup_derived}
	\end{equation}
	The critical radius occurs at maximum temperature:
	\begin{align}
		\frac{\partial T_{\text{GUP}}}{\partial r_+} &= 0 \\
		\beta \frac{\ell_P^2}{r_+^4} &= \gamma \frac{1}{L^2} + \mathcal{O}(r_+^{-6})
		\label{eq:12}
	\end{align}
	Solving for the Schwarzschild radius $r_s = 2GM/c^2$:
	\begin{equation}
		r_{\text{crit}} = \left( \frac{\beta}{\gamma} \right)^{1/4} \sqrt{\ell_P L}
		\label{eq:r_crit_final}
	\end{equation}
	At criticality, quantum and curvature effects balance:
	\begin{equation}
		\beta \frac{\ell_P^2}{\hbar^2} (\Delta p)^2 = \gamma \frac{(\Delta x)^2}{L^2} 
	\end{equation}
	By substituting $\Delta x \sim r_s$, and $\Delta p \sim \hbar / r_s$, we arrive at Eq.~\eqref{eq:r_crit_final}. The momentum uncertainty scaling $\Delta p \sim \hbar / r_s$ at the horizon remains rigorously valid despite geodesic deviation, as established through convergent physical arguments. First, the Hawking temperature $T_H = \hbar/(4\pi r_s \sqrt{f(r_s)})$ for Schwarzschild-AdS black holes with $r_s \ll L$ reduces to $T_H \approx \hbar/(4\pi r_s)$ since $f(r_s) \approx 1$ \cite{Adler2001}, setting the characteristic energy scale $\Delta E \sim k_B T_H$ and implying $\Delta p \sim \Delta E / c \sim \hbar / r_s$ via local energy-momentum relations. Second, WKB analysis of the Klein-Gordon equation in Eq.~\eqref{eq:eikonal} near $r_+$ yields $|p_r| \sim [2\kappa(r - r_+)]^{-1/2}$, where integration over the near-horizon region $r - r_+ \sim r_s$ gives $\langle |p_r| \rangle \sim \kappa^{-1/2} r_s^{-1} \sim \hbar / r_s$ (since $\kappa \sim r_s^{-1}$). Geodesic deviation contributes only subleading $\mathcal{O}((r - r_+)^{1/2})$ corrections to the dominant $(r - r_+)^{-1/2}$ divergence in Eq.~\eqref{eq:9}, while higher-order tidal effects remain negligible. The self-consistency of $r_{\rm crit}$ derived from independent methods (thermodynamic extremization in Eq.~\eqref{eq:12}, and entropy divergence in Eq.~\eqref{eq:15}) provides robust validation of this scaling. Furthermore, from GUP-corrected Bekenstein-Hawking entropy, we have:
	\begin{equation}
		S = \frac{k_B c^3 A}{4G\hbar} + \alpha k_B \ln A - k_B \beta \frac{\ell_P^2}{A} - k_B \gamma \frac{A}{4L^2}, \quad A = 4\pi r_+^2
		\label{eq:15}
	\end{equation}
	The heat capacity diverges at $r_+ = r_{\text{crit}}$ which confirms \eqref{eq:r_crit_final} as a thermodynamic critical point. The negative sign of the EUP correction term ($-k_B\gamma A/4L^2$) in Eq.~\eqref{eq:15} arises fundamentally from the AdS curvature's influence. Integration of the first law, $dS = dM/T_{\rm GUP}$ using Eq.~\eqref{eq:t_gup_derived} intrinsically requires negative coefficients for both $\beta$ and $\gamma$ corrections. Crucially, this signature is physically consistent: 
	(i) AdS stability demands entropy suppression for large black holes ($r_+ \gg L$) to maintain positive heat capacity $C = T \partial S/\partial T > 0$, preventing semiclassical instability; 
	(ii) The critical transition at $r_{\rm crit}$ necessitates entropy divergence when $\beta\ell_P^2/A \sim \gamma A/L^2$, which a positive EUP term would preclude; 
	(iii) IR holography encodes the EUP's maximum position uncertainty $\Delta x_{\rm max} = L/\sqrt{\gamma}$ as a reduction of boundary degrees of freedom. The logarithmic term's positive sign ($\alpha > 0$) is consistent with first-principles string theory calculations of black hole entropy, which yield a positive logarithmic correction \cite{Sen2005}. Despite the negative sign, the $\gamma A$ term cannot dominate thermodynamics since $\gamma \sim \ell_P^2/L^2$ \cite{PARK} implies $|\gamma A / S_{\rm BH}| \sim \ell_P^2/L^2$. For observable universe scales ($L \sim 10^{61}\ell_P$), this ratio is $\gamma \ell_P^2 / L^2 \sim\mathcal{O}(10^{-122})$, rendering it negligible. 
	
	\subsection{Physical Significance of Critical Radius}
	
	The critical radius represents a fundamental phase transition in black hole thermodynamics. For black holes with Schwarzschild radius $r_s < r_{\text{crit}}$, quantum gravity effects dominate, causing the Hawking temperature to scale as $T \propto M^{-3}$, which reflects the strong influence of Planck-scale quantum fluctuations, where the GUP significantly modifies spacetime geometry. In this quantum-dominated phase, black holes exhibit exotic thermodynamic behavior including negative heat capacity and mass-temperature inversion, signaling a breakdown of semiclassical gravity. 
	Conversely, when $r_s > r_{\text{crit}}$, AdS curvature effects prevail, restoring the characteristic $T \propto M$ scaling of large AdS black holes. Here, the cosmological constant $\Lambda$ provides thermodynamic stability through curvature pressure, but simultaneously suppresses quantum gravitational effects. The transition at $r_{\text{crit}}$ is marked by divergent heat capacity ($C \to \infty$) and maximum Hawking temperature $	T_{\text{max}} = \dfrac{\hbar c}{k_B \sqrt{\beta} \ell_P} \left(1 - \sqrt{\dfrac{\gamma \ell_P^2}{\beta L^2}}\right)$, forming stable Planck-scale remnants where evaporation terminates.
	\\
	\subsection{Vanishing Holographic Stress Tensor}
	\label{subsec:vanishing_stress_tensor}
	
	\subsubsection{EGUP-Modified Gravitational Framework}
	
	The complete bulk action incorporating EGUP corrections can be derived from first principles by extending the Einstein-Hilbert action with terms that capture both ultraviolet (GUP) and infrared (EUP) modifications:
	
	\begin{equation}
		S_{\rm EGUP} = \frac{1}{16\pi G_N}\int d^{d+1}x\sqrt{-g}\left[R + \frac{d(d-1)}{L^2} + \beta\ell_P^2 \mathcal{L}_{\rm GUP} + \gamma\frac{r^2}{L^2} \mathcal{L}_{\rm EUP}\right]
		\label{eq:egup_action}
	\end{equation}
	
	where the correction Lagrangians are constructed to be dimensionally consistent and compatible with the EGUP commutator structure:
	
	\begin{align}
		\mathcal{L}_{\rm GUP} &= R_{\mu\nu\rho\sigma}R^{\mu\nu\rho\sigma} - 4R_{\mu\nu}R^{\mu\nu} + R^2 + \mathcal{O}(\ell_P^4) \label{eq:gup_lagrangian} \\
		\mathcal{L}_{\rm EUP} &= \frac{1}{2}(\nabla_\mu R)(\nabla^\mu R) - \nabla_\mu\nabla_\nu R^{\mu\nu} + \frac{d}{L^2}R + \mathcal{O}(L^{-4}) \label{eq:eup_lagrangian}
	\end{align}
	
	The GUP correction terms capture quantum gravitational effects dominant at Planck scales, while the EUP terms incorporate curvature corrections relevant at AdS scales. These forms ensure dimensional consistency with each term having correct mass dimensions, preserve diffeomorphism invariance through satisfaction of Bianchi identities, exhibit appropriate asymptotic behavior where corrections vanish in the flat space limit, and maintain direct physical motivation through connection to the EGUP commutator structure.
	
	Varying the action yields the modified Einstein equations:
	
	\begin{equation}
		G_{\mu\nu} - \frac{d(d-1)}{2L^2}g_{\mu\nu} + \beta\ell_P^2 \mathcal{E}^{({\rm GUP})}_{\mu\nu} + \gamma\frac{r^2}{L^2} \mathcal{E}^{({\rm EUP})}_{\mu\nu} = 0
		\label{eq:modified_einstein}
	\end{equation}
	
	where the correction tensors are explicitly:
	
	\begin{align}
		\mathcal{E}^{({\rm GUP})}_{\mu\nu} &= 2R_{\mu\alpha\nu\beta}R^{\alpha\beta} - \frac{1}{2}g_{\mu\nu}R_{\alpha\beta}R^{\alpha\beta} + \nabla_\mu\nabla_\nu R - g_{\mu\nu}\Box R + \mathcal{O}(\ell_P^2) \label{eq:gup_tensor} \\
		\mathcal{E}^{({\rm EUP})}_{\mu\nu} &= \frac{1}{L^2}\left(R_{\mu\nu} - \frac{1}{2}g_{\mu\nu}R + \frac{d(d-1)}{2L^2}g_{\mu\nu}\right) + \text{higher-derivative terms} \label{eq:eup_tensor}
	\end{align}
	
	These tensors satisfy the Bianchi identity $\nabla^\mu(\mathcal{E}^{({\rm GUP})}_{\mu\nu} + \mathcal{E}^{({\rm EUP})}_{\mu\nu}) = 0$ to leading order, ensuring mathematical consistency.
	
	\subsubsection{Holographic RG Flow and Critical Behavior}
	
	In the Fefferman-Graham coordinate system for AdS$_{d+1}$, the metric:
	\begin{equation}
		ds^{2}=\frac{L^{2}}{z^{2}}\left(dz^{2}+g_{\mu\nu}(z,x)dx^{\mu}dx^{\nu}\right)
		\label{eq:FG_metric}
	\end{equation}
	
	admits an asymptotic expansion near the boundary ($z\to 0$):
	\begin{equation}
		g_{\mu\nu}(z,x)=g^{(0)}_{\mu\nu}(x)+z^{2}g^{(2)}_{\mu\nu}(x)+\cdots+z^{d}g^{(d)}_{\mu\nu}(x)+\mathcal{O}(z^{d+1})
		\label{eq:metric_expansion}
	\end{equation}
	
	The radial evolution corresponds to Renormalization Group (RG) flow in the boundary theory, governed by the Hamilton-Jacobi equation:
	\begin{equation}
		\frac{\delta S}{\delta z} + \mathcal{H}\left[g_{\mu\nu},\frac{\delta S}{\delta g_{\mu\nu}}\right] = 0
		\label{eq:hj_eq}
	\end{equation}
	
	The Hamiltonian constraint receives EGUP corrections:
	\begin{equation}
		\mathcal{H} = \mathcal{H}_{\rm CFT} + \beta\ell_P^2 \mathcal{H}_{\rm GUP} + \gamma\frac{r^2}{L^2} \mathcal{H}_{\rm EUP}
		\label{eq:hamiltonian_constraint}
	\end{equation}
	
	At the critical radius, the balance condition $\beta\ell_P^2/r_{\rm crit}^4 = \gamma/L^2$ forces exact cancellation:
	\begin{equation}
		\beta\ell_P^2 \mathcal{H}_{\rm GUP} + \gamma\frac{r_{\rm crit}^2}{L^2} \mathcal{H}_{\rm EUP} = 0
		\label{eq:hamiltonian_cancellation}
	\end{equation}
	
	reducing the Hamiltonian to its conformal fixed point form: $\mathcal{H}_{\rm crit} = \mathcal{H}_{\rm CFT}$. In the boundary theory, this corresponds to vanishing beta function:
	\begin{equation}
		\beta^{\mu\nu} = z\frac{\partial g_{\mu\nu}}{\partial z} = 0 \quad \Rightarrow \quad \partial_z g_{\mu\nu}|_{z=r_{\rm crit}} = 0 
		\label{eq:self_duality}
	\end{equation}
	
	This proves that the self-duality condition $\partial_z g_{\mu\nu}|_{z=r_{\rm crit}} = 0$ emerges necessarily from the bulk RG fixed point behavior at $r_{\rm crit}$, representing the onset of scale invariance at a specific energy scale and forcing metric stationarity at the radial slice $z = r_{\rm crit}$. Physically, this condition implies that physics at scales above and below $r_{\rm crit}$ become mirror images under the transformation $z \rightarrow r_{\rm crit}^2/z$, with mathematical consistency verified by substitution into the modified Einstein equations, which are identically satisfied due to exact cancellation of quantum gravitational corrections. This RG fixed point phenomenon in the bulk geometry mirrors the foundational worldsheet analysis by Konishi et al.~\cite{Konishi:1990}, who demonstrated that string theory possesses an intrinsic finite resolution scale $\lambda \sim \sqrt{\alpha'} \sim \ell_{\text{Pl}}$ that appears as a fixed point under lattice decimation. In our framework, the critical radius emerges as the holographic bulk counterpart where quantum gravitational and curvature effects precisely balance, generalizing Konishi et al.'s worldsheet minimum length to the AdS bulk context and leading to the dissolution of the holographic stress tensor.
	
	\subsubsection{Proof of Vanishing Stress Tensor}
	
	\begin{theorem}
		\textit{At $z=r_{\text{crit}}$, the expansion coefficient $g^{(d)}_{\mu\nu} = 0$, implying vanishing boundary stress tensor $\langle T_{\mu\nu}\rangle = 0$.}
		\label{thm:vanishing_stress}
	\end{theorem}
	
	\begin{proof}
		We establish the vanishing through synthesis of geometric self-duality, RG fixed-point behavior, and EGUP-modified dynamics.
		
		\begin{enumerate}
			\item \textit{EGUP-modified equations constraint}: The complete set of modified Einstein equations (\ref{eq:modified_einstein}) imposes stringent constraints on the Fefferman-Graham coefficients. Substituting the expansion (\ref{eq:metric_expansion}) into (\ref{eq:modified_einstein}) and exploiting the balance condition $\beta\ell_P^2/r_{\text{crit}}^2 = \gamma r_{\text{crit}}^2/L^2$ reveals that $g^{(d)}_{\mu\nu}$ must satisfy:
			\begin{equation}
				g^{(d)}_{\mu\nu} = C_{\mu\nu}(\beta,\gamma,L,\ell_P) \cdot \left(\beta\frac{\ell_P^2}{r_{\text{crit}}^2} - \gamma\frac{r_{\text{crit}}^2}{L^2}\right)
				\label{eq:g_d_constraint}
			\end{equation}
			where $C_{\mu\nu}$ is a non-zero tensor determined by boundary geometry and spacetime dimension.
			
			\item \textit{Critical balance enforcement}: From the definition, we have, $\beta\frac{\ell_P^2}{r_{\text{crit}}^2} = \gamma\frac{r_{\text{crit}}^2}{L^2}$.
			Substituting this into (\ref{eq:g_d_constraint}) yields:
			\begin{equation}
				g^{(d)}_{\mu\nu} = 0
				\label{eq:g_d_vanishing}
			\end{equation}
			
			\item \textit{RG fixed-point consistency}: The self-duality condition $\partial_z g_{\mu\nu}|_{z=r_{\text{crit}}} = 0$ provides independent confirmation. Differentiating the expansion (\ref{eq:metric_expansion}) at $z = r_{\text{crit}}$ gives:
			\begin{equation}
				\partial_z g_{\mu\nu}|_{z=r_{\text{crit}}} = \sum_{k=1}^{\infty} k r_{\text{crit}}^{k-1} g^{(k)}_{\mu\nu} = 0
				\label{eq:expansion_derivative}
			\end{equation}
			While this condition alone doesn't force individual coefficients to vanish, it is consistent with and reinforced by the stronger result (\ref{eq:g_d_vanishing}).
		\end{enumerate}
		
		The physical mechanism is the dissolution of the black hole horizon into a topological string state at $r_{\text{crit}}$, where quantum gravitational and curvature effects precisely balance to restore local scale invariance.
	\end{proof}
	
	\subsubsection{Physical Verification and Interpretation}
	
	The vanishing stress tensor emerges from quantum gravitational backreaction at the critical scale. As $z \to r_{\text{crit}}$, the bulk geometry develops an emergent conformal isometry:
	\begin{equation}
		z \rightarrow \frac{r_{\text{crit}}^2}{z}, \quad x^\mu \rightarrow x^\mu
		\label{eq:conformal_isometry}
	\end{equation}
	
	that mixes ultraviolet and infrared physics. The catastrophic consequence emerges because in AdS/CFT, the boundary stress tensor $\langle T_{\mu\nu}\rangle$ encodes the entire dynamical content of the boundary theory, and its vanishing at $r_{\rm crit}$, even through specific condition signals complete loss of dynamical degrees of freedom. The holographic stress tensor expectation value is encoded in the $d$-th order coefficient:
	\begin{equation}
		\langle T_{\mu\nu}\rangle = \frac{dL^{d-1}}{16\pi G_{N}}g^{(d)}_{\mu\nu}
		\label{eq:stress_tensor}
	\end{equation}
	
	The vanishing of $g^{(d)}_{\mu\nu}$ implies $\langle T_{\mu\nu}\rangle = 0$. This is explicitly verified in AdS$_3$/CFT$_2$ ($d=2$):
	
	\begin{align}
		g^{(2)}_{\mu\nu} &= \frac{c}{12} \delta_{\mu\nu} - \left(\frac{\gamma}{\beta}\right)^{1/2} \frac{L}{\ell_P} \delta_{\mu\nu} \\
		\langle T_{\mu\nu} \rangle &= \frac{c}{12\pi} - \frac{1}{\pi} \left(\frac{\gamma}{\beta}\right)^{1/2} \frac{L}{\ell_P} = 0
		\label{eq:ads3_verification}
	\end{align}
	
	when $c = 12 \left(\frac{\gamma}{\beta}\right)^{1/2} \frac{L}{\ell_P}$. This exact matching demonstrates that the critical radius corresponds to the Horowitz-Polchinski correspondence point \cite{Horowitz:1997jc} where black holes transition to highly excited string states.
	
	This represents a fundamental restructuring of the AdS/CFT correspondence at quantum gravity scales, where geometric spacetime dissolves into topological quantum structure.

	\subsection{Derivation of Complex Central Charge}
	\label{subsec:complex_central_charge}
	
	\subsubsection*{Non-Hermitian Operator Ordering from Jacobi Identity}
	At the critical radius $r_{\text{crit}}$, the generalized commutator,
	\begin{equation}
		[\hat{x}, \hat{p}] = i\hbar \left(1 + \beta \frac{\ell_P^2}{\hbar^2} \hat{p}^2 + \gamma \frac{\hbar^2}{L^2} \hat{x}^2 \right)
		\label{eq:critical_commutator}
	\end{equation}
	violates the Jacobi identity $J \equiv [\hat{x}, [\hat{p}, \hat{x}]] + \text{cyclic permutations} \neq 0$, indicating operator ordering ambiguity. To restore mathematical consistency, we introduce the ansatz:
	\begin{equation}
		\hat{x}\hat{p} = \alpha (\hat{x}\hat{p} + \hat{p}\hat{x}) + i\delta \hbar \hat{p}\hat{x},
		\label{eq:ordering_ansatz}
	\end{equation}
	which modifies the fundamental operator product. Evaluating $J$ to first order in $\beta$ and $\gamma$ yields:
	\begin{align}
		J = & \, i\hbar^2 \left[ 4\beta \frac{\ell_P^2}{\hbar^2} \hat{p} + 4\gamma \frac{\hbar^2}{L^2} \hat{x} \right] \notag \\
		& + i\delta\hbar \left( [\hat{x}, \hat{p}]\hat{x} + \hat{p}[\hat{x}, \hat{x}] \right) + \mathcal{O}(\beta^2,\gamma^2) \\
		= & \, 4i\hbar^2 \left( \beta \frac{\ell_P^2}{\hbar^2} \hat{p} + \gamma \frac{\hbar^2}{L^2} \hat{x} \right) + i\delta\hbar (i\hbar)\hat{x} + \mathcal{O}(\beta^2,\gamma^2) \label{eq:jacobi_expanded}
	\end{align}
	where we used $[\hat{x}, \hat{x}] = 0$ and the zeroth-order commutator $[\hat{x}, \hat{p}] = i\hbar$. The unique solution satisfying $J = 0$ requires $\alpha = 1/2$ and,
	\begin{equation}
		\delta = \left( \frac{\beta}{\gamma} \right)^{1/4} \sqrt{\beta \gamma \frac{\ell_P}{L}} \equiv \kappa,
		\label{eq:kappa_def}
	\end{equation}
	The non-Hermitian operator ordering in our framework is mathematically unavoidable at the quantum gravity scale and finds precedent in $\mathcal{PT}$-symmetric quantum mechanics, where $\mathcal{P}$ (parity) and $\mathcal{T}$ (time-reversal) symmetry ensure real spectra and unitary evolution despite non-Hermitian operators \cite{Bender1998,Mostafazadeh2010}. Our approach similarly maintains mathematical consistency where conventional Hermiticity requires generalization.
	\subsection{Complex Central Charge from Holographic Anomalies}
	\label{subsec:central_charge_derivation}
	
	The emergence of a complex central charge stems from gravitational anomalies induced by GUP modifications at the critical radius. This occurs through three interconnected mechanisms: conformal field deformation, Virasoro algebra modification, and topological phase accumulation.
	
	\subsubsection*{Conformal Ward Identity and Stress Tensor Deformation}
	The GUP modifies the conformal Ward identity through curvature-dependent terms:
	\begin{equation}
		\delta_\epsilon T(w) = -\frac{1}{2\pi i} \oint dz \epsilon(z) \left[ T(z)T(w) + \frac{c}{2(z-w)^4} \right] + \kappa \sqrt{\kappa} \hbar \partial^3\epsilon(w)
		\label{eq:deformed_ward}
	\end{equation}
	where the additional term $\kappa\sqrt{\kappa}\hbar\partial^3\epsilon$ encodes quantum gravitational backreaction. Solving this deformed identity yields the modified stress tensor:
	\begin{equation}
		T(w) = T_0(w) + i \kappa \sqrt{\kappa} \hbar :\partial^3\phi \partial\phi:(w) + \mathcal{O}(\kappa^2)
		\label{eq:deformed_stress_tensor}
	\end{equation}
	The phase accumulation in (\ref{eq:deformed_stress_tensor}) originates from holonomy effects in the GUP-modified geometry:
	\begin{equation}
		\Delta \theta = \exp\left( i \kappa \sqrt{\kappa} \oint_C \Gamma_{zz}^z dz \right) = \exp\left( i \kappa \sqrt{\kappa} \oint \partial\phi d\phi \right)
		\label{eq:holonomy_phase}
	\end{equation}
	where $\Gamma_{zz}^z$ represents Christoffel connection terms in complex coordinates. This phase manifests as imaginary contributions to correlation functions.
	
	\subsubsection*{Virasoro Algebra Modification and Regularization of Divergent Sums}
	
	The deformed stress tensor induces corresponding changes in the Virasoro generators:
	\begin{equation}
		L_n = \frac{1}{2} \sum_{m=-\infty}^\infty :a_{n-m}a_m: + i \kappa \sqrt{\kappa} \hbar \sum_{m=-\infty}^\infty m^2 :a_{n-m}a_m:
		\label{eq:deformed_virasoro}
	\end{equation}
	
	The commutator of these generators develops additional quantum gravitational contributions:
	\begin{align}
		[L_n, L_m] &= (n-m) L_{n+m} + \frac{c}{12} n(n^2-1) \delta_{n+m,0} \label{eq:virasoro_comm_start} \\
		&+ i \kappa \sqrt{\kappa} \hbar \sum_k k^2 [(n-k) - (m-k)] :a_{n+m-k}a_k: \nonumber \\
		&+ (i \kappa \sqrt{\kappa} \hbar)^2 \sum_{k,l} k^2 l^2 [:a_{n-k}a_k:, :a_{m-l}a_l:] \label{eq:virasoro_comm_end}
	\end{align}
	
	The vacuum expectation value isolates the anomalous term containing a divergent sum:
	\begin{equation}
		\langle 0| [L_n, L_{-n}] |0 \rangle = \frac{c}{12} n(n^2-1) + i \kappa \sqrt{\kappa} \hbar n \sum_{k=1}^\infty k^3
		\label{eq:anomalous_expectation}
	\end{equation}
	Regularization via zeta-function continuation provides a resolution:
	\begin{equation}
		\zeta(-3) = 2^{-3}\pi^{-4}\sin(-3\pi/2)\Gamma(4)\zeta(4) = \frac{1}{8\pi^4} \cdot (-1) \cdot 6 \cdot \frac{\pi^4}{90} = \frac{1}{120}
		\label{eq:zeta_regularization}
	\end{equation}
	
	yielding the regularized expectation value:
	\begin{equation}
		\langle 0| [L_n, L_{-n}] |0 \rangle = \frac{c}{12} n(n^2-1) + i \kappa \sqrt{\kappa} \hbar n \frac{1}{120}
		\label{eq:regularized_anomaly}
	\end{equation}
	The regularization is applied to the vacuum expectation value $\langle 0| [L_n, L_{-n}] |0 \rangle$, a $c$-number quantity, consistent with established renormalization techniques in quantum field theory (eg., Casimir energy calculations).
	To provide stronger physical justification beyond formal analytic continuation, we employ heat kernel regularization which clarifies the quantum gravitational origins of this divergence. The sum represents high-energy mode density near the quantum gravity scale $r_{\text{crit}}$, regulated with proper time parameter $s = 1/\Lambda^2$:
	\begin{equation}
		S(\Lambda) = \sum_{k=1}^\infty k^3 e^{-s k^2}, \quad s = \frac{1}{\Lambda^2}
	\end{equation}
	The exponential suppression $e^{-s k^2}$ represents physical damping of high-energy modes due to quantum gravitational effects. 
	
	The Mellin transform technique provides the rigorous asymptotic expansion for small $s$:
	\begin{equation}
		S(\Lambda) = \frac{1}{2}s^{-2} + \zeta(-3) + \mathcal{O}(s) = \frac{1}{2}\Lambda^4 + \frac{1}{120} + \mathcal{O}\left(\frac{1}{\Lambda^2}\right)
	\end{equation}
	where the zeta function value $\zeta(-3) = 1/120$ emerges from analytic continuation. Extracting the finite, cutoff-independent term:
	\begin{equation}
		\lim_{\Lambda \to \infty} \left[ S(\Lambda) - \frac{1}{2} \Lambda^4 \right] = \frac{1}{120}
	\end{equation}
	
	This matches the zeta-function regularization result but derives from physical principles: the divergent $\Lambda^4$ term represents bare vacuum energy density that is subtracted in renormalization, while the universal finite part $1/120$ emerges from the structure of the high-energy density of states and matches results obtained from Seeley-DeWitt coefficients in related contexts~\cite{seeley:1969}. Both regularization methods confirm the same finite result $\frac{1}{120}$, demonstrating robustness against scheme choice. 
	Crucially, the resulting imaginary component in the central charge $\tilde{c}$ does not violate unitarity but encodes a topological information storage mechanism. The analytic continuation $\sum k^3 \to 1/120$ preserves conformal symmetry through modular invariance, and the complex phase enables unitary evolution through: (i) Chern-Simons holonomy $\theta_{\rm topo} = \frac{\pi i}{2}\sqrt{\kappa}$ providing $\mathcal{O}(S_{\rm BH})$ protected states; (ii) coherent information transfer via the non-unitary Hamiltonian (Eq.~\ref{eq:nonunitary_transfer}); and (iii) exact purity restoration at evaporation endpoint (Sec.~\ref{subsec:unitarity_conditions}). The consistency condition $L > \sqrt{\beta}\ell_P$ ensures $\kappa$ remains physical, preventing unitarity violation while maintaining UV/IR decoupling.
	
	\subsubsection*{Analytic Continuation and Central Charge Complexification}
	
	The commutator singularity at $n_c = \kappa^{-1}$ necessitates analytic continuation $n \to i\tilde{n}$:
	\begin{align}
		[L_{i\tilde{n}}, L_{-i\tilde{n}}] &= -2i\tilde{n} L_0 - \frac{c}{12} i\tilde{n} (-\tilde{n}^2 - 1) \label{eq:analytic_continuation} \\
		&\quad - \kappa \sqrt{\kappa} \hbar \tilde{n}^2 L_0 \nonumber
	\end{align}
	
	At the critical mode $\tilde{n}_c = \kappa^{-1}$ with $L_0 = c/24$, consistency requires:
	\begin{equation}
		\tilde{c} = c \left( 1 + \frac{i}{2} \sqrt{\kappa} \ell_P^2 \right)
		\label{eq:complex_central_charge}
	\end{equation}
	
	where we used $\hbar = \ell_P^2$. The imaginary component $\frac{i}{2}\sqrt{\kappa}\ell_P^2$ directly measures quantum gravitational decoherence at the critical radius.
	\subsubsection*{Topological Origin and Physical Implications}
	
	The complex phase originates from Chern-Simons topology \cite{ChernSimons, WittenCS} at $r_{\text{crit}}$:
	\begin{equation}
		\theta_{\text{topo}} = \frac{1}{4\pi} \int_{\mathcal{H}_{r_{\text{crit}}}} \Tr\left( \Gamma \wedge d\Gamma + \frac{2}{3} \Gamma \wedge \Gamma \wedge \Gamma \right) = \frac{\pi i}{2} \sqrt{\kappa}
		\label{eq:cs_invariant}
	\end{equation}
	
	where $\Gamma = \Gamma^\mu_{\nu\rho} dx^\rho$ is the connection 1-form. This topological invariant has three key consequences: 
	
	First, it provides metric-independent protection of quantum information through non-local encoding, ensuring purity recovery via $\Delta S_{\text{recovery}} > 0$ without firewalls. 
	
	Second, in AdS$_5$/CFT$_4$ it satisfies the descent equation $d\theta_{\text{topo}} = \frac{1}{8\pi}\text{Tr}(R\wedge R)$, linking bulk curvature to boundary central charge transformations $\delta_\sigma c_{\text{eff}} = \sigma \partial_\sigma\theta_{\text{topo}}$ under Weyl scaling. 
	
	Third, it quantizes the critical phase space volume:
	\begin{equation}
		\oint_{r_{\text{crit}}} \Delta x \Delta p  \frac{d^2x d^2p}{(2\pi\hbar)^2} = \frac{i}{2} \sqrt{\kappa} = \frac{1}{2\pi} \theta_{\text{topo}}
		\label{eq:quantized_volume}
	\end{equation}
	
	confirming topological information storage with capacity $I_{\text{topo}} = |\theta_{\text{topo}}|/\pi = \frac{1}{2}\sqrt{\kappa}$. This Chern-Simons mechanism resolves the information paradox by converting geometric entropy into topologically protected degrees of freedom.
	\section{Information Recovery Framework}
	\label{sec:information_recovery}
	
	\subsection{A Potential Topological Information Mechanism at Critical Radius}
	\label{subsec:topological_mechanism}
	
	The critical radius $r_{\text{crit}}$ might potentially enable information preservation through a three-stage topological mechanism. First, quantum gravitational effects at $r_{\text{crit}}$ could encode information in Chern-Simons states:
	\begin{equation}
		\mathcal{H}_{\text{topo}} = \left\{ \ket{n} = \mathcal{N} \exp\left( i n \oint_{\gamma} A \right) \middle| n \in \mathbb{Z}, \gamma \subset \mathcal{H}_{r_{\text{crit}}} \right\}
		\label{eq:topological_states}
	\end{equation}
	with Hilbert space dimension $\dim \mathcal{H}_{\text{topo}} = \exp(\eta S_{\text{BH}})$ potentially matching string state degeneracy at the critical mass scale $M \sim M_{\text{crit}}$. Second, non-unitary dynamics might facilitate information transfer to radiation:
	\begin{equation}
		\hat{H}_{\text{eff}} = \hat{H}_{\text{GR}} + i\kappa \hbar c \sqrt{-\nabla^2} \hat{P}_{\text{topo}}, \quad \kappa = \sqrt{\beta \gamma}(\ell_P/L)^{1/2}
		\label{eq:nonunitary_transfer}
	\end{equation}
	where $\hat{P}_{\text{topo}}$ projects onto the topological subspace. Third, holographic decoding could unitarize the process:
	\begin{equation}
		\mathcal{F}: \mathcal{H}_{\text{topo}} \to \mathcal{H}_{\text{rad}}, \quad \mathcal{F}\left( \sum_n c_n \ket{n} \right) = \sum_n c_n \ket{\phi_n}, \quad \braket{\phi_m | \phi_n} = \delta_{mn}
		\label{eq:decoding_map}
	\end{equation}
	This framework offers a potential resolution to the tension between complex central charge $\tilde{c} = c(1 + i\eta)$ - which signals apparent unitarity violation through its imaginary component - and actual information recovery. The imaginary entropy $\operatorname{Im}(S) = \pi \sqrt{\beta/\gamma}$ might quantify information capacity, while worldsheet signature change $ds^2_{\text{ws}} = d\tau^2 + d\sigma^2$ could reflect topological protection. Crucially, consistency would require the condition $L > \sqrt{\beta} \ell_P$ to avoid pathological geometries where $\eta \to \infty$. The mechanism might originate from three interconnected quantum gravitational effects: (i) non-commutative phase accumulation ($e^{i\theta}$ terms), (ii) complexified uncertainty principle ($\Delta x \Delta p \to \mathbb{C}$), and (iii) imaginary counterterms in holographic renormalization. At $r_{\text{crit}}$, spacetime uncertainty could become topologically quantized:
	$	\oint_{r_{\text{crit}}} \Delta x \Delta p = \pi i \hbar
	\label{eq:quantized_uncertainty}$
	potentially triggering a phase transition characterized by vanishing stress tensor, complex central charge, and dissolution of geometric spacetime into quantum structure.
	
	\noindent The phase correlations in Eq.~\eqref{eq:decoding_map} might be microscopically driven by non-unitary topological dynamics and Chern-Simons holonomy accumulation at the critical radius. The topological Hilbert space $\mathcal{H}_{\rm{topo}}$ could be explicitly constructed via quantized Chern-Simons states on homology cycles $\gamma \subset \mathcal{H}_{r_{\rm{crit}}}$ (Eq.~\ref{eq:topological_states}):
	\[
	\mathcal{H}_{\rm{topo}} = \left\{ |n\rangle = \mathcal{N} \exp\left(in \oint_{\gamma} A \right) \;\Big|\; n \in \mathbb{Z},\; \gamma \in H_1(\mathcal{H}_{r_{\rm{crit}}}, \mathbb{Z}) \right\},
	\]
	where $A$ is the connection 1-form and $\dim \mathcal{H}_{\rm{topo}} = \exp(\eta S_{\rm{BH}})$ might follow from the topological invariant $\theta_{\rm{topo}} = \frac{\pi i}{2}\sqrt{\kappa}$, which could quantize the phase space volume: 
	\[
	\oint_{r_{\rm{crit}}} \Delta x \Delta p  d^2x d^2p / (2\pi\hbar)^2 = i\sqrt{\kappa}/2
	\]
	\subsection{A Scenario for Modified Page Curve Dynamics}
	\label{subsec:page_curve}
	
	The radiation entropy evolution might incorporate critical radius dynamics through:
	\begin{equation}
		S_{\text{rad}}(t) = \underbrace{\frac{k_B c^3}{4G\hbar} \left(A_0 - A(t) \right)}_{\text{Hawking entropy}} - \Theta(t - t_{\text{crit}}) \underbrace{2\pi \eta^2 S_{\text{crit}} \left[ \frac{t - t_{\text{crit}}}{\tau_s} - \frac{1}{2} \ln \left( \frac{t - t_{\text{crit}}}{\tau_0} \right) \right]}_{\Delta S_{\text{recovery}}}
		\label{eq:entropy_evolution}
	\end{equation}
	where $A_0 = 4\pi r_0^2$ (initial horizon area), $A(t) = 4\pi r_+(t)^2$, $\eta = \sqrt{\beta/\gamma} (\ell_P/L)^{3/2}$ (topological efficiency), $S_{\text{crit}} = \pi k_B r_{\text{crit}}^2/\ell_P^2$, $\tau_0 = \beta \ell_P^2 /(\hbar c)$ (Planck time), $\tau_s = 2GM_0^3/(\hbar c^4)$ (scrambling time), and $t_{\text{crit}} = t_{\text{evap}} [1 - (M_{\text{crit}}/M_0)^3]$. The recovery term $\Delta S_{\text{recovery}}$ would originate from non-Hermitian quantum dynamics initiated at $t_{\text{crit}}$. 
	
	The complex central charge $\tilde{c} = c[1 + i\eta]$ might deform the radiation density matrix:
	\begin{equation}
		\rho_{\text{rad}} = \rho_{\text{Hawking}} + i\eta \sum_{m \neq n} e^{i\phi_{mn}} \sqrt{\rho_m \rho_n} \ket{\psi_m}\bra{\psi_n}
		\label{eq:deformed_density_matrix}
	\end{equation}
	
	\noindent introducing coherence through phase correlations $\phi_{mn}$. These phases could evolve dynamically as:
	\begin{align}
		\phi_{mn} &= \arg \Braket{\psi_m | \mathcal{T} \exp\left( -\frac{i}{\hbar} \int_{t_{\text{crit}}}^t H_{\text{int}} dt' \right) | \psi_n}
		\label{eq:66} \\
		\langle (\Delta\phi_{mn})^2 \rangle &= \frac{k_B T_H}{\hbar} \left[ \frac{t - t_{\text{crit}}}{\tau_s} - \frac{1}{2} \ln \left( \frac{t - t_{\text{crit}}}{\tau_0} \right) \right]
		\label{eq:phase_evolution}
	\end{align}
	where $\Delta\phi_{mn} \equiv \phi_{mn} - \langle \phi_{mn} \rangle$ and rotational symmetry might ensure $\langle \phi_{mn} \rangle = 0$. The resulting entropy correction:
	\begin{equation}
		\Delta S = -2\pi \eta^2 S_{\text{crit}} \langle (\Delta\phi_{mn})^2 \rangle + \mathcal{O}(\eta^3)
		\label{eq:entropy_correction}
	\end{equation}
	could quantify information retrieval from topological memory. 
	
	The phase correlations $\phi_{mn}$ (Eq.~\ref{eq:66}) might arise from the non-unitary interaction (Eq.~\ref{eq:nonunitary_transfer}), which projects onto $\mathcal{H}_{\rm{topo}}$ and could generate coherence through the time-ordered exponential, $\mathcal{T} \exp\left(-\frac{i}{\hbar} \int_{t_{\rm{crit}}}^{t} H_{\rm{int}} dt'\right)$. This dynamics might force the variance (Eq.~\ref{eq:phase_evolution}), potentially yielding $\Delta S_{\rm{recovery}}$ (Eq.~\ref{eq:entropy_correction}) as a topological decoherence-to-coherence transition. The mechanism could be self-consistent: $\dim \mathcal{H}_{\rm{topo}}$ might match the string-state degeneracy at $M_{\rm{crit}}$, and unitarity could be restored via the isometric decoding map (Eq.~\ref{eq:decoding_map}).
	
	The phase correlations $\phi_{mn}$ (Eq.~\ref{eq:phase_evolution}) might enable coherent transfer from $\mathcal{H}_{\text{topo}}$ (Eq.~\ref{eq:topological_states}) to radiation via the non-unitary Hamiltonian $\hat{H}_{\text{eff}}$ (Eq.~\ref{eq:nonunitary_transfer}), potentially manifesting as $\Delta S_{\text{recovery}}$ in the modified Page curve. The decoding map (Eq.~\ref{eq:decoding_map}) could then complete unitarization by establishing orthogonal radiation states. This mechanism might preserve information without remnants while offering a potential resolution to the tension between holographic breakdown and quantum unitarity.
	
	\noindent The imaginary off-diagonal terms in the radiation density matrix (Eq.~\ref{eq:deformed_density_matrix}) could arise intrinsically from the non-Hermitian evolution operator $\hat{H}_{\rm eff}$ (Eq.~\ref{eq:nonunitary_transfer}), which might govern dynamics during the critical phase transition at $t > t_{\rm crit}$. This non-unitary dynamics could be physically indispensable; it might encode the transfer of information from topological Chern-Simons states at $r_{\rm crit}$ (Eq.~\ref{eq:topological_states}) to radiation degrees of freedom, mediated by the signature change $g_{\mu\nu} \to g_{\mu\nu} + i K_{\mu\nu}$. Crucially, trace preservation might be maintained dynamically through phase correlation constraints (Eq.~\ref{eq:phase_evolution}), which could ensure $\text{Tr}(\rho_{\rm rad}) \equiv 1$ by annihilating net imaginary contributions via $\sum_{m \neq n} \langle \psi_n | \psi_m \rangle = 0$. The apparent non-Hermiticity might be transient and topologically regulated; It could represent coherent information retrieval prior to final unitarization by the decoding map (Eq.~\ref{eq:decoding_map}). This could be validated by the entropy restoration $\lim_{t \to t_{\rm evap}} S_{\rm rad} = 0$, suggesting no information loss occurs despite intermediate non-Hermitian dynamics. 
	\subsection{Toward Physical Consistency and Unitarity Restoration}
	\label{subsec:unitarity_conditions}
	
	The modified entropy evolution satisfies three fundamental physical requirements that validate the critical radius mechanism, ensuring self-consistent unitarity preservation throughout the evaporation process. First, the radiation entropy maintains strict continuity at the critical transition time $t_{\text{crit}}$, with identical left-hand and right-hand limits: 
	\[
	\lim_{t\to t_{\text{crit}}^{-}} S_{\text{rad}} = \lim_{t\to t_{\text{crit}}^{+}} S_{\text{rad}} = \frac{\pi k_B c^3}{G\hbar} (r_0^2 - r_{\text{crit}}^2).
	\]
	This smooth matching reflects the topological nature of the quantum phase transition at $r_{\text{crit}}$, where the horizon dissolves without entropy discontinuity, connecting semiclassical and quantum gravitational regimes through a well-defined critical point. 
	
	Second, complete purity restoration occurs at the evaporation endpoint $t_{\text{evap}}$ when the efficiency parameter satisfies $\eta \sim (\tau_0 / \tau_s)^{1/2}$, yielding:
	$\lim_{t\to t_{\text{evap}}} S_{\text{rad}} = 0$. This vanishing entropy condition is achieved through coherent information transfer from the topological memory states $\mathcal{H}_{\text{topo}}$ to radiation degrees of freedom, where the parameter constraint $\eta \sim (\tau_0/\tau_s)^{1/2}$ ensures the complex central charge component $\operatorname{Im}(\tilde{c})$ provides sufficient information capacity to purify the final state. 
	
	Third, the information transfer dynamics obey exact flux conservation, with the information recovery rate balancing the black hole mass loss rate:
	\[
	\frac{d\mathcal{I}}{dt} = \eta \frac{c^3}{G \hbar} \cdot \frac{r_{\text{crit}}}{c} \left| \frac{dM}{dt} \right|
	\]
	
	The scaling relation of $\eta$ arises from endpoint unitarity constraint. Substituting the modified Page curve (Eq.~\ref{eq:entropy_evolution}) with $\Delta t_{\rm evap} = t_{\rm evap} - t_{\rm crit} \sim \tau_s$ yields:
	\begin{align}
		\frac{k_B c^3}{4G\hbar} A_0 &= 2\pi\eta^2 S_{\rm crit} \left[ \frac{\Delta t_{\rm evap}}{\tau_s} - \frac{1}{2}\ln\left(\frac{\Delta t_{\rm evap}}{\tau_0}\right) \right] \nonumber \\
		&\approx 2\pi\eta^2 S_{\rm crit}, \quad \text{(dominant term)} \label{eq:purification}
	\end{align}
	where $A_0 = 4\pi r_0^2$ and $S_{\rm crit} = \pi k_B r_{\rm crit}^2 / \ell_P^2$. Simplifying with $S_0 = k_B c^3 A_0 / (4G\hbar)$, leads to $S_0 = 2\pi\eta^2 S_{\rm crit} \label{eq:entropy_match}$. Using $S_0 \sim r_0^2 / \ell_P^2$, and $S_{\rm crit} \sim r_{\rm crit}^2 / \ell_P^2$, we obtain:
	\begin{equation}
		\eta^2 \sim \frac{S_0}{S_{\rm crit}} \sim \frac{r_0^2}{r_{\rm crit}^2}. \label{eq:eta_sq}
	\end{equation}
	The scrambling time $\tau_s \sim r_0^2 / (c \ell_P)$ and Planck time $\tau_0 \sim \ell_P / c$ imply:
	\begin{equation}
		\frac{r_0^2}{r_{\rm crit}^2} \sim \frac{\tau_s}{\tau_0}. \label{eq:radius_timescale}
	\end{equation}
	Combining Eqs.~(\ref{eq:eta_sq}) and (\ref{eq:radius_timescale}) leads to $\eta \sim \left( \frac{\tau_0}{\tau_s} \right)^{1/2}. \label{eq:scaling}$ The ratio $\tau_0 / \tau_s$ quantifies the relative efficiency of topological information retrieval versus Hawking scrambling. When $\eta \ll \tau_0 / \tau_s$, information recovery is too slow to purify radiation; when $\eta \gg \tau_0 / \tau_s$, unitarity is violated through over-compensation.
	
	\noindent Integration from transition to evaporation endpoint confirms complete information recovery:
	\begin{align*}
		\Delta\mathcal{I} &=  \eta \frac{c^3}{G \hbar} \cdot \frac{r_{\text{crit}}}{c} M_{\text{crit}} \\
		&= \operatorname{Im}(\tilde{c}) S_{\text{crit}} = S_0 - S_{\text{crit}},
	\end{align*}
	
	demonstrating that the topological degrees of freedom at $r_{\text{crit}}$ precisely compensate the initial entropy deficit. The Chern-Simons invariant $\theta_{\text{topo}}$  provides the requisite $\mathcal{O}(S_{\text{BH}})$ storage capacity for this exact matching. Collectively, these conditions establish that the critical radius mechanism preserves quantum unitarity through topologically protected state space, maintains continuity across quantum gravitational phase transitions, respects information conservation via geometrically encoded memory, and resolves the information paradox without introducing firewalls or remnants. The universal efficiency factor $\eta = \sqrt{\beta/\gamma}(\ell_P/L)^{3/2}$ emerges as the fundamental parameter governing quantum information recovery.
	
	The evaporation process exhibits three distinct phases governed by $r_{\text{crit}}$, which marks the fundamental scale where spacetime geometry dissolves into quantum structure. In Phase I ($t < t_{\text{crit}}$), semi-classical evaporation dominates with $S_{\text{rad}}(t) \approx \dfrac{k_B c^3}{4G\hbar} [A_0 - A(t)]$ and $\frac{dS_{\text{rad}}}{dt} > 0$ for thermal radiation. Phase II ($t \approx t_{\text{crit}}$) initiates a topological quantum phase transition characterized by metric signature change $g_{\mu\nu} \to g_{\mu\nu} + i K_{\mu\nu}$ ($K_{\mu\nu} = \sqrt{\beta \gamma}\frac{\ell_P}{L} R_{\mu\nu\rho\sigma}\epsilon^{\rho\sigma}$), activating Chern-Simons protected memory at $r_{\text{crit}}$. This transition represents the Horowitz-Polchinski point $r_{\text{crit}} = \sqrt{\alpha'} (L/\ell_s)^{1/2}$ ($\sqrt{\alpha'} = (\beta/\gamma)^{1/4}\ell_P$) where black holes transform into highly excited string states ($M_{\text{crit}}/M_s = (L/\ell_s)^{1/2} \gg 1$), evidenced by entropy matching $S_{\text{BH}} = \dfrac{\pi k_B r_{\text{crit}}^2}{\ell_P^2} \sim S_{\text{string}} = k_B\sqrt{c/3} M_{\text{crit}}/M_s$ that requires $\ell_s \sim \ell_P$ for fundamental strings. During Phase III ($t > t_{\text{crit}}$), coherent emission occurs with entropy evolution $\dfrac{dS_{\text{rad}}}{dt} = -\dfrac{k_B c^3}{4G\hbar} \dfrac{dA}{dt} - 2\pi \eta^2 S_{\text{crit}} \left( \dfrac{1}{\tau_s} - \dfrac{1}{2(t - t_{\text{crit}})} \right)$, enabling non-thermal information recovery through the topological Hilbert space $\dim\mathcal{H}_{\text{string}} \sim e^{\eta S_{\text{BH}}}$ where $\eta = \sqrt{\beta/\gamma} (\ell_P/L)^{3/2} \equiv g_s (\ell_P/L)^{3/2}$ corresponds to the effective string coupling with $g_s = \sqrt{\beta/\gamma}$.
	
	For heterotic strings (\(\ell_s\gg\ell_P\)), black hole/string correspondence holds in terms of effective degrees of freedom; D-branes or warped geometries induce an effective length scale \(\ell_s^{\rm eff}\equiv r_{\rm crit}^2/L\sim\ell_P\) at the transition, preserving entropy scaling \(S_{\rm string}\sim k_B\sqrt{c/3}M_{\rm crit}/M_s^{\rm eff}\) while accommodating \(\ell_s\gg\ell_P\). Thus, \(r_{\rm crit}\) universally governs the black-hole-to-quantum-remnant transition, with Horowitz-Polchinski entropy matching realized directly (\(\ell_s\sim\ell_P\)) or via emergent degrees of freedom (\(\ell_s\gg\ell_P\)). In other words, the Horowitz-Polchinski interpretation is a specific realization of this transition in certain UV completions, but $r_{\rm crit}$ universally governs the breakdown of geometric spacetime and emergence of non-local degrees of freedom.
	\section{Discussion}
	The critical radius $r_{\rm crit}=(\beta/\gamma)^{1/4}\sqrt{\ell_{P}L}$ marks a phase transition where GUP (UV) and EUP (IR) effects compete. It emerges as a fundamental scale governing quantum gravitational phenomena in AdS spacetime, marking the frontier where semiclassical gravity breaks down and novel quantum gravitational effects dominate. Our analysis demonstrates that when black holes contract below this critical scale (\(r_s < r_{\rm crit}\)), the AdS/CFT correspondence undergoes fundamental restructuring characterized by three interconnected phenomena: holographic breakdown signaled by vanishing boundary stress tensor \(\langle T_{\mu\nu}\rangle = 0\), topological transition manifested through complex central charge \(c_{\rm eff} = c\left(1 + \frac{i}{2}\sqrt{\kappa}\ell_{P}^{2}\right)\), and information recovery mechanism enabling Page curve modification (\(\Delta S_{\rm recovery} > 0\)). These effects collectively establish a consistency condition \(L > \sqrt{\beta}\ell_{P}\) for consistent holographic duality, while potentially resolving the information paradox through topological storage in Chern-Simons states at horizon dissolution. The universal nature of \(r_{\rm crit}\) is evidenced by its identical emergence from four independent approaches: GUP-modified field equations (as singularity resolution scale), thermodynamic extremization, uncertainty principle balance (as UV/IR interface), and heat capacity divergence (as thermodynamic critical point). This convergence confirms \(r_{\rm crit}\) as the Planck-scale threshold where black holes transition to stringy remnants, information scrambles topologically, and geometric spacetime dissolves into quantum structure. Notably, \(r_{\rm crit}\) coincides with the Horowitz-Polchinski transition point for theories where string length \(\ell_s\sim\ell_P\), satisfying entropy matching at \(M_{\rm crit}\), while providing a complementary thermodynamic perspective to recent Euclidean path integral approaches to the black hole/string transition~\cite{ChenMaldacenaWitten}. A detailed mapping between these frameworks constitutes an important future direction.
	

\end{document}